\documentclass[11pt,oneside]{article}

\usepackage{palatino}
\usepackage[margin=1in]{geometry}
\usepackage{color,xcolor,soul}
\usepackage[abspage,user,savepos]{zref}
\usepackage[colorlinks, citecolor=blue]{hyperref}  
\usepackage{natbib}
\usepackage{fp}
\usepackage{framed}
\usepackage{cuted,balance}
\usepackage{setspace}

\usepackage{graphicx} 
\usepackage{epsfig} 
\usepackage{mathtools}
\usepackage{dsfont,amsmath} 
\usepackage{amssymb}  
\usepackage{algorithm}
\usepackage{algpseudocode}
\usepackage[capitalize]{cleveref}

\newcommand{\be}{\begin{equation}}
\newcommand{\ee}{\end{equation}}
\newcommand{\Expected}{\mathrm{E}}
\newcommand{\nn}{\nonumber}
\newcommand{\beliefPi}{\hat{\pi}}
\newcommand{\beliefQ}{\hat{Q}}
\newcommand{\beliefV}{\hat{v}}
\DeclareMathOperator*{\argmax}{argmax}

\newcounter{definition}
\newenvironment{definition}{\refstepcounter{definition}\par\medskip
   \noindent 
   \textbf{Definition \thedefinition.} \em \rmfamily}{\medskip}

\newcounter{assumption}
\newenvironment{assumption}{\refstepcounter{assumption}\par\medskip
   \noindent  
   \textbf{Assumption \theassumption.} \em \rmfamily}
{\medskip}

\newcounter{theorem}
\newenvironment{theorem}{\refstepcounter{theorem}\par\medskip
   \noindent  
   \textbf{Theorem \thetheorem.} \em \rmfamily}
{\medskip}

\newenvironment{proof}{
   \noindent \textit{Proof:} \rmfamily}{\hfill $\square$\medskip}

\begin{document}

\date{}

\title{\LARGE \bf
On the Global Convergence of Stochastic Fictitious Play \\ in Stochastic Games with Turn-based Controllers\thanks{This is a corrected version of \citep{ref:Sayin22CDC}.}
}

\author{Muhammed O. Sayin
\thanks{M. O. Sayin is with Department of Electrical and Electronics Engineering, Bilkent University, Ankara T\"{u}rkiye.
        {\tt\small sayin@ee.bilkent.edu.tr}}%
}
\maketitle

\bigskip

\begin{center}
\textbf{Abstract}
\end{center}
This paper presents a learning dynamic with almost sure convergence guarantee for any stochastic game with turn-based controllers (on state transitions) as long as stage-payoffs induce a zero-sum or identical-interest game. Stage-payoffs for different states can even have different structures, e.g., by summing to zero in some states and being identical in others. The dynamics presented combines the classical stochastic fictitious play with value iteration for stochastic games. There are two key properties: (i) players play finite horizon stochastic games with increasing lengths within the underlying infinite-horizon stochastic game, and (ii) the turn-based controllers ensure that the auxiliary stage-games (induced from the continuation payoff estimated) are strategically equivalent to zero-sum or identical-interest games.

\begin{spacing}{1.245}

\section{Introduction}
Stochastic games, introduced by \cite{ref:Shapley53}, can model interactions among multiple intelligent (non-cooperative) decision-makers with long-term (possibly different) objectives in dynamic multi-state environments with Markovian state transitions and stage-payoffs. This makes stochastic games an ideal model for multi-agent reinforcement learning with non-cooperative agents. 

Characterization of equilibrium for stochastic games has already been studied extensively \citep{ref:Neyman03}. Many methods have been proposed to compute an equilibrium for stochastic games, including the inaugural work of \cite{ref:Shapley53} as a generalization of value iteration for Markov decision processes to two-player zero-sum stochastic games and its model-free version Minimax-Q algorithm \citep{ref:Littman94} and so on.

An important justification for the predictive power of game-theoretical equilibrium analysis is that an equilibrium can arise naturally as an outcome of non-equilibrium adaption of learning agents \citep{ref:Fudenberg98}. For example, fictitious play (and its variants) is a popular stylized learning model for agents forming beliefs about opponent strategies (as if they are stationary) and choosing greedy responses to these beliefs in the repeated play of the same game \citep{ref:Fudenberg98}. It is well-established that if agents follow stochastic fictitious play dynamics in strategic-form games played repeatedly, the dynamics converges to a \textit{Nash distribution} in important classes of strategic-form games such as games with an interior evolutionary stable strategy (ESS), zero-sum games, potential games, and super-modular games \citep{ref:Hofbauer02}.\footnote{Nash distribution corresponds to Nash equilibrium associated with smoothly perturbed payoffs \citep{ref:Hofbauer05}.} 

However, the existing results on whether non-equilibrium adaptation of learning agents can converge to equilibrium or not in stochastic games has been very limited until recently. A comprehensive answer to this question will strengthen the predictive power of equilibrium analysis and its applications in stochastic games, and therefore, is of particular interest for multi-agent reinforcement learning with frontier applications in artificial intelligence and autonomous systems.

For example, new variants of fictitious play (or best response) dynamics combining classical fictitious play (or best response dynamics) with $Q$-learning (or value iteration) have been shown to converge in certain classes of stochastic games such as zero-sum and identical-payoff \citep{ref:Sayin20,ref:Leslie20,ref:Baudin21}. However, whether these dynamics can converge to equilibrium in general classes of stochastic games is an open problem. Examples include stage-payoffs with structures strategically equivalent to zero-sum or identical payoff, or stage-payoffs with state-dependent structures (e.g., stage-payoffs sum to zero in some states while they are identical in others). Furthermore, even when players have completely opposite stage-payoffs, the underlying stochastic game would not be zero sum if the players have different discount factors.

In this paper, I present a learning dynamic with global convergence guarantees for any general-sum stochastic game in which stage-payoffs induce a zero-sum or identical-interest game and for each state there exists a certain single player controlling the state transitions. Note that turn-based stochastic games in which players take actions in turn (e.g., as in board games) and stochastic games with single-controller in which only a single player controls every state transitions are special cases of stochastic games with turn-based controllers. 

The key property of the dynamics is that players play the underlying stochastic game by dividing the infinite horizon into epochs with finitely many stages. In each epoch, they play a finite horizon version of the underlying game. This provides players opportunities to revise and improve their beliefs formed about the opponent strategies and the underlying game dynamics. In that respect, the dynamics presented has a similar flavor with the near-optimal learning algorithm introduced in \citep{ref:Kearns02} for single-agent reinforcement learning and its generalization to two-player zero-sum stochastic games in \citep{ref:Brafman02}. In another related paper, \cite{ref:Perolat18} focused on actor-critic algorithms in the repeated play of multi-stage games where each state can get visited only once in each repetition. The dynamics presented differs from these results by addressing the cases where stage-payoff functions can have different structures at different states in addition to player-dependent discount factors.

The dynamics presented also has several advantages compared to the recent results on learning in stochastic games such as \citep{ref:Sayin20,ref:Leslie20,ref:Baudin21}. For example, stage-payoffs specific to different states can have different structures (from completely competitive to identical ones) and players can have different discount factors. However, in \citep{ref:Sayin20,ref:Leslie20,ref:Baudin21}, all stage-payoff functions either sum to zero or they are identical, and all players have the same discount factor. Furthermore, the results in \citep{ref:Sayin20,ref:Leslie20,ref:Baudin21} can not be extended to the cases with stage-payoffs that are strategically equivalent to zero-sum or identical payoff cases.

The paper is organized as follows: I describe stochastic games and auxiliary stage-game framework in Section \ref{sec:game} and present the new variant of stochastic fictitious play dynamics for stochastic games in Section \ref{sec:play}. Section \ref{sec:main} includes the convergence result and the proof. I discuss several future research directions in Section \ref{sec:extension} and conclude the paper with some remarks in Section \ref{sec:conclusion}. 

\section{Stochastic Games}\label{sec:game}

Consider a dynamic game with finitely many states and finitely many players. Players play this game over discrete time $k=0,1,\ldots$ by taking actions simultaneously and receiving the corresponding \textit{stage-payoffs} at each stage $k$. Across stages, the state of the game can change causally depending on the history of the play while the players have the long-term objective of maximizing the amount of stage-payoffs received over infinite horizon. Stochastic games (also known as \textit{Markov} games) is such a multi-state dynamic game in which state transitions and stage-payoffs have Markov property by depending only on the \textit{current} state and \textit{current} action profile. 

Formally, an $n$-player infinite-horizon stochastic game can be characterized by a tuple $\langle S, A, r, p, \gamma \rangle 
$ in which $S$ is the set of \textit{finitely} many states, $A$ is a collection of (state-invariant) action sets $A^i$ for each $i\in [n]$ and $s\in S$, i.e., $A= \prod_{i\in [n]} A^i$, and  $r$ is a collection of stage-payoff functions $r^i : S\times A\rightarrow \mathbb{R}$ for each $i\in [n]$.\footnote{The formulation and results can be generalized to the state-variant action sets rather straightforwardly. Furthermore, we define $[n] := \{1,\ldots,n\}$.} For any triple $(s,a,s')$, the transition probability function $p(s'|s,a)$ gives the probability that the game transits to state $s'$ from state $s$ when players play the action profile $a=(a^i)_{i\in [n]} \in A$ at the current state. Set Player $i$ as the typical player. Then, Player $-i$ denotes Player $i$'s opponents, i.e., $-i:=\{j\in[n]:j\neq i\}$. The objective of Player $i$ is to maximize the expected sum of the stage-payoffs discounted with $\gamma^i \in [0,1)$. Note that we let players have different discount factors.

Suppose that players can observe the state and the actions of the opponents with perfect recall. Define the observations that Player $i$ makes up to stage $k$ by $h_k := \{s_0,a_0,\ldots,s_{k-1},a_{k-1},s_k\}$, where $s_l$ and $a_l$ denote, respectively, the state and action profile at stage $l$. We also let players independently randomize their actions according to a behavioral strategy, denoted by $\pi^i=(\pi_k^i\in \Delta(A^i))_{k\geq 0}$, that says which action should be taken with what probability, e.g., based on $h_k$ at stage $k$.\footnote{We denote the simplex of probability distributions over $A^i$ by $\Delta(A^i)$.} Correspondingly, given the strategy profile $\pi=(\pi^i)_{i\in [n]}$, the objective (also known as utility) of Player $i$ is given by
\be\label{eq:utility}
U^i(\pi):=\Expected\left[\sum_{k=0}^{\infty} (\gamma^i)^k r^i(s_k,a_k)\right],
\ee 
where the expectation is taken with respect to the randomness on $(s_k,a_k)$ for each $k\geq 0$. 

Note that \cite{ref:Shapley53} (and later \cite{ref:Fink64}) showed that in two-player zero-sum (and $n$-player general-sum) stochastic games, there always exists a \textit{stationary} mixed-strategy equilibrium in which players' strategies are stationary by depending only on the current state (and not depending on time). However, here we do not restrict the strategies of the players to the stationary ones.

\begin{definition}\label{def:turnbased}
We say that a stochastic game has \textit{turn-based controller} if for each state $s$ there exists a single player $i_s$ such that the state transition probability $p(s'|s,a)$ for any $s'$ depends only on $(s,a^{i_s})$, i.e., does not depend on other players' actions $a^{-i_s}$.
\end{definition}

\subsection{Auxiliary Stage-Game Framework}
If there were only one state and the discount factor $\gamma = 0$, then stochastic games would reduce to strategic-form (also known as normal-form or one-shot) games. 
Hence, we can view the stage-wise interactions among players at stage $k$ as a strategic-form \textit{auxiliary stage-game} specific to the current state $s$ and stage $k$. We call payoff functions of these auxiliary stage-games by \textit{$Q$-function}, denoted by $Q_k^i:S\times A \rightarrow \mathbb{R}$ for all $k$. In other words, the auxiliary stage-game of stage $k$ specific to state $s$ can be characterized by the pair $G_k(s):=\langle A, Q_k(s,\cdot)\rangle$. 

The $Q$-function $Q_k^i(s,a)$ gives the expected utility of Player $i$ associated with $(s,a)$ at stage $k$, given by\footnote{This is indeed a scaled version of the payoff function by $\gamma^{-k}$ without loss of generality.}
\begin{align}
Q_k^i(s,a) :&= \Expected\left[\sum_{l=k}^{\infty}(\gamma^i)^{l-k}r^i(s_l,a_l)\; \big|\; (s_k,a_k) = (s,a)\right]\nn\\
&= r^i(s,a) + \gamma^i \sum_{s'\in S}p(s'|s,a) v_{k+1}^i(s'),\quad\forall a,\label{eq:Q}
\end{align}
where we define $v_{k+1}^i:S\rightarrow\mathbb{R}$ by
\begin{align}
v_{k+1}^i(s') :=&\; \Expected\left[\sum_{l=k+1}^{\infty}(\gamma^i)^{l-(k+1)}r^i(s_l,a_l)\; \big|\; s_{k+1}=s'\right],\label{eq:v}
\end{align}
which can also be written as
\be
v_{k+1}^i(s') = \Expected\left[Q_{k+1}^i(s',a)\right].
\ee
In \eqref{eq:Q} and \eqref{eq:v}, the expectations are taken with respect to the randomness on the future state and action profiles $(s_l,a_l)_{l>k}$. However, the randomness on $(s_l,a_l)_{l>k}$ is \textit{ambiguous} at stage $k$ by depending on how players would take (and randomize) their actions in the future stages $l>k$. On the other hand, there is also another ambiguity about which immediate stage-payoff a player would get and how state changes since both of them depend on the action profile of all players and the players act simultaneously. 

In the next section, we discuss how players can resolve these uncertainties through a simple learning model in which they form beliefs about these uncertainties while acting greedily and independently.

\section{Stochastic Fictitious Play}\label{sec:play}

In this section, we will first take a detour to describe stochastic fictitious play for the repeated play of a strategic-form game. Then, I will present a new variant extending it to stochastic games, based on the auxiliary stage-game framework.

\subsection{Stochastic Fictitious Play for Repeated Games}\label{sec:sfp}

Consider that a strategic-form game $G=\langle A, r\rangle$ played \textit{repeatedly} over discrete time $k=0,1,\ldots$. At stage $k$, players take actions simultaneously and receive payoffs associated with the current action profile. 
Each player's objective is to maximize her utility rather than equilibrium play. However, this objective is not a well-defined optimization problem by depending on other players' actions. Correspondingly, each player can reason about how others would play in the current repetition of the game based on how they played in the past. 

Fictitious play is a well-studied best-response type learning dynamics in game theory. Players form a belief about the mixed strategy of each opponent based on the history of the play under a possibly erroneous assumption that the opponent is playing according to a stationary strategy and they act myopically by taking the best response against the belief formed \citep{ref:Fudenberg98}. For example, Player $i$ can form a belief $\hat{\pi}_k^j \in \Delta(A^j)$ about Player $j$'s strategy based on the empirical average of the actions taken by Player $j$ until stage $k$, which can be written as
\be\label{eq:emp}
\hat{\pi}_k^j = \frac{1}{k} \sum_{l=0}^{k-1} a_k^j,\quad \forall k>0,
\ee 
where we consider the action taken $a_k^j$ as a pure (or degenerate mixed) strategy, i.e., $a_k^j\in \Delta(A^j)$.

\textit{Stochastic (also known as smoothed) fictitious play} is a variant of the classical fictitious play in which each player still forms a belief as in \eqref{eq:emp} but randomize her response as if her expected utility is given by
$
\Expected[r^i(a)] + \tau \eta^i(\pi^i),
$ 
where $\tau>0$ is a temperature parameter controlling the scale of the perturbation, and $\eta^i:\Delta(A^i)\rightarrow\mathbb{R}$ is a smooth and strictly concave function whose gradient is unbounded at the boundary of the simplex $\Delta(A^i)$ \citep{ref:Fudenberg98,ref:Hofbauer02, ref:Hofbauer05}. This perturbation penalizes strategies assigning small probabilities to some actions. Furthermore, it leads to a unique (smoothed) best response given by 
\be\label{eq:smoothed}
B^i(\pi^{-i}) = \argmax_{\mu^i\in \Delta(A^i)} \left\{\Expected_{(a^i,a^{-i})\sim (\mu^i,\pi^{-i})} [r^i(a^i,a^{-i})] + \tau \eta^i(\mu^i)\right\},
\ee
where $\pi^{-i} = (\pi^j)_{j\neq i}$.
For example, the entropy as a smooth perturbation, 
$$
\eta^i(\mu^i) = - \sum_{a^i} \mu^i(a^i) \log(\mu^i(a^i)),
$$ 
leads to the logit choice function \citep{ref:Hofbauer05}. 

The discrete-time stochastic fictitious play dynamics for the beliefs formed can be written as
\be
\hat{\pi}_{k+1}^j = \hat{\pi}_k^j + \alpha_k(a_k^j - \hat{\pi}_k^j),\quad\forall  k>0,
\ee
and $\hat{\pi}_1^j = a_1^j$ for $j\in[n] $, where $a_k^j \sim B^j(\hat{\pi}_k^{-j})$ and $\alpha_k\in [0,1]$ is a vanishing step size so that players can give more (or less) weight on their belief than the observed action compared to the empirical average. For example, $\alpha_k = 1/(k+1)$ leads to the empirical average \eqref{eq:emp}.

\begin{definition}
We say that a strategy profile $\pi_*$ is a \textit{Nash distribution} provided that
\be\label{eq:NE}
\pi_*^i = B^i(\pi_*^{-i}),
\ee 
for all $i\in [n]$.
\end{definition}

In other words, any player does not have any incentive to change her strategy unilaterally with respect to her smoothly perturbed utility \citep{ref:Fudenberg98}. Note also that $\pi_*$ is an $\epsilon$-Nash equilibrium where $\epsilon \rightarrow 0$ as $\tau \rightarrow 0$ \citep{ref:Hofbauer05}.

\subsection{Stochastic Fictitious Play for Stochastic Games}

Players are playing an auxiliary stage-game specific to a state when the associated state is visited. However, the payoff of these stage-games, i.e., the $Q$-function, is not necessarily stationary different from the repeated game framework. Furthermore, the payoff is ambiguous by depending on players' strategies in future stages. Players can attempt to resolve this ambiguity by forming a belief about the $Q$-function as they do for the opponent strategy in the fictitious play. For example, players can form a belief about the $Q$-function by assuming that each opponent plays according to a behavioral strategy while the player herself always plays according to her smoothed best response (based on her beliefs). 

However, how players would act in an auxiliary stage-game depends on when the associated state is visited since the behavioral strategy depends on the history of the play. Therefore, players may not be able to form beliefs that can be revised and improved with new observations, e.g., by taking their empirical average as in \eqref{eq:emp}. To address this ambiguity, I present a new learning scheme in which players play finite-horizon versions of the underlying stochastic games repeatedly within the underlying game while the length of the horizon grows in time. This framework lets players revise and improve their beliefs with new observations made. 

To this end, we partition the infinite horizon into epochs indexed by $t=1,2,\ldots$. Epoch $t$ consists of $t$ substages indexed by $\ell=1,\ldots,t$, and therefore, corresponds to a finite-horizon version of the underlying stochastic game that can be also be characterized by the tuple $\langle S,A,r,p,\gamma\rangle$. However, given the strategy profile $\pi^{(t)} = \{\pi^{(t)}_{\ell}\}_{\ell=1}^t$, the utility of Player $i$ associated with epoch $t$ is now given by
\be
U^{(t),i}(\pi^{(t)}) :=\Expected \left[\sum_{\ell = 1}^{t} (\gamma^i)^{\ell-1}r^i(s_{\ell}^{(t)},a_{\ell}^{(t)})\right], 
\ee
where $(s_{\ell}^{(t)},a_{\ell}^{(t)})$ corresponds to the state and action profile realized at substage $\ell$ of epoch $t$ and the expectation is taken with respect to the randomness on the state and action profile.
Correspondingly, at substage $\ell$ of epoch $t$, we can view players as playing a strategic-form \textit{auxiliary substage-game} specific to the current state $s$ and substage $\ell$. This substage-game can be characterized by the pair $G_{\ell}^{(t)}(s):=\langle A, Q_{\ell}^{(t)}(s,\cdot)\rangle$, in which the payoff function is given by
\begin{align}
Q_{\ell}^{(t),i}(s,a) &= \Expected\left[\sum_{l=\ell}^{t} (\gamma^i)^{l-\ell} r^i(s_{l}^{(t)},a_{l}^{(t)}) \;|\; (s_{\ell}^{(t)}, a_{\ell}^{(t)}) = (s,a) \right],\nn\\
&= r^i(s_{\ell}^{(t)},a_{\ell}^{(t)}) + \gamma^i \sum_{s'\in S} p(s'|s,a) v_{\ell +1}^{(t),i}(s'),\quad \forall a,\label{eq:Qsub}
\end{align}
for all $\ell=1,\ldots,t-1$, and the continuation payoff can be written as
\be
v_{\ell+1}^{(t),i}(s') = \Expected\left[ Q_{\ell+1}^{(t),i}(s',a) \right],\label{eq:vsub}
\ee
and $Q_{t}^{(t),i}(s,a)=r^i(s,a)$.

\begin{figure}[t!]
\centering
\includegraphics[width=.55\textwidth]{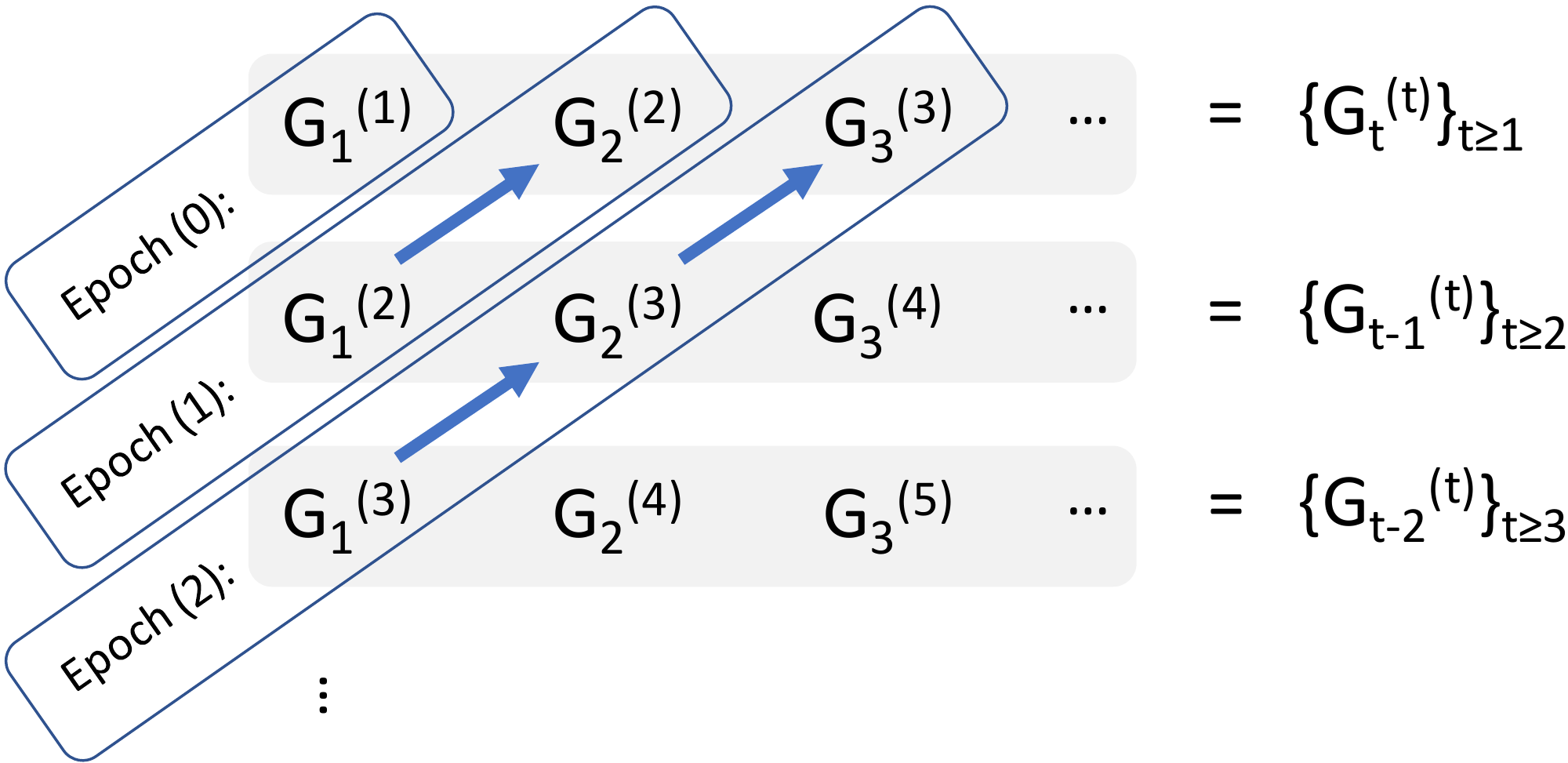}
\caption{A figurative illustration for the repeated play of auxiliary substage-games $G_{\ell}^{(t)}$ at epoch $t$ and substage $\ell$. For example, $\{G_{t-m}^{(t)}\}_{t>0}$ corresponds to the substage-games that are $m$ substage before the last one within the associated epoch.  Arrows represent the order these substage-games played in an epoch.}\label{fig:model}
\end{figure}

Through this learning scheme, players can now revise and improve the beliefs formed about the payoff of these substage-games across epochs because each $Q_{t-m}^{(t),i}$, for $t\geq m$, corresponds to the payoff of the substage-game that is $m$ substage before the last one in each epoch $t$, as illustrated in Fig. \ref{fig:model}. For example, $Q_{t}^{(t),i}$ corresponds to the last one. Hence, as in the classical fictitious play, the players can assume that opponents play according to stationary strategies in the repetitions of the substage-game $G_{t-m}^{(t)}(s)$ for $t\geq m$. 

Player $i$ forms beliefs $\beliefPi_{\ell}^{(t),j}(s)$, for each $j\neq i$, and $\beliefQ_{\ell}^{(t),i}(s,\cdot)$ that, respectively, correspond to Player $j$'s strategy, for $j\neq i$, and the payoff at the auxiliary substage-game $G_{\ell}^{(t)}(s)$. Across the repetitions of the associated substage-game $G_{t-m}^{(t)}(s)$, for $t\geq m$, players can revise and improve their beliefs based on new observations made, e.g., by taking the empirical average of actions as in \eqref{eq:emp}. Player $i$ also keeps track of $\beliefV_{\ell}^{(t)}(s)$ corresponding to the expected payoff Player $i$ gets in the auxiliary substage-game $G_{\ell}^{(t)}(s)$, as described in \eqref{eq:vsub}, where $v_+^i(s')$ corresponds to the continuation payoff for state $s'$ starting from the next (sub)stage.

\begin{table}[t]
\normalsize
\caption{The $k$th Iteration of Stochastic Fictitious Play Update}\label{algo:SFP}
\begin{algorithmic}
\hrule
\Procedure{$\mathrm{SFP}_k$}{$Q^i,\pi^{-i}$}
\State $B^i \gets \argmax_{\mu^i}\left\{\Expected_{(a^i,a^{-i})\sim (\mu^i,\pi^{-i})}[Q^i(a)] + \tau \eta^i(\mu^i)\right\}$
\State Play action $a^i_*\sim B^i$
\State $a^j_* \gets$ action of Player $j$, for all $j\neq i$ 
\State $\pi^{j}\gets \pi^{j} + \alpha_k (a^j_* - \pi^j)$, for all $j\neq i$
\State $v^i \gets \Expected_{(a^i,a^{-i})\sim (B^i,\pi^{-i})}[Q^i(a)]$
\State\Return $(\pi^{-i},v^i)$
\EndProcedure
\hrule
\end{algorithmic}
\end{table}

\begin{table}[t]
\normalsize
\caption{The $k$th Iteration of Q-function Update for State $s$}\label{algo:Q}
\begin{algorithmic}
\hrule
\Procedure{$\mathrm{QUpdate}_k$}{$s,Q^i(s,\cdot),v_+^i$}
\State$Q^i(s,\cdot) \gets Q^i(s,\cdot) + \beta_k \left(r^i(s,\cdot) + \gamma^i \sum_{s'}p(s'|s,\cdot) v_+^i(s') - Q^i(s,\cdot)\right)$
\State\Return $Q^i(s,\cdot)$
\EndProcedure
\hrule
\end{algorithmic}
\end{table}

\begin{table}[t]
\caption{Stochastic Fictitious Play for Stochastic Games}
\normalsize
\label{algo}
\hrule
\begin{algorithmic}[1]
\For{Each epoch $t=0,1,\ldots$}
\State $\beliefPi_1^{(t),j}(s,a) \gets 1/|A^j|$, for all $(s,a)$, $j\neq i$
\State $\beliefQ_1^{(t),i}(s,a) \gets r^i(s,a)$, for all $(s,a)$
\For{Each substage $\ell=1,2,\ldots,t$}
\If{$s$ is the current state}
\State $(\beliefPi_{\ell+1}^{(t+1),-i}(s),\beliefV_{\ell+1}^{(t+1),i}(s))\gets \mathrm{SFP}_{c_{\ell}^{(t)}(s)}(\beliefQ_{\ell}^{(t),i}(s,\cdot),\beliefPi_{\ell}^{(t),-i}(s))$
\If{$\ell<t$}
\State 
$\beliefQ_{{\ell}+1}^{(t+1),i}(s,\cdot)\gets \mathrm{QUpdate}_{c_{\ell}^{(t)}(s)} (s,\beliefQ_{\ell}^{(t),i}(s,\cdot),\beliefV_{\ell+1}^{(t),i})$
\Else
\State $\beliefQ_{\ell+1}^{(t+1),i}(s,\cdot) \gets r^i(s,a)$\EndIf
\State $c_{\ell+1}^{(t+1)}(s)\gets c_{\ell}^{(t)}(s)+1$
\Else
\State $(\beliefPi_{\ell+1}^{(t+1),-i}(s),\beliefV_{\ell+1}^{(t+1),i}(s))\gets (\beliefPi_{\ell}^{(t),-i}(s),\beliefV_{\ell}^{(t),i}(s))$
\State $\beliefQ_{\ell+1}^{(t+1),i}(s,\cdot) \gets \beliefQ_{\ell}^{(t),i}(s,\cdot)$
\State $c_{\ell+1}^{(t+1)}(s)\gets c_{\ell}^{(t)}(s)$
\EndIf
\EndFor
\EndFor
\end{algorithmic} 
\hrule
\end{table}

Player $i$ follows the stochastic fictitious play update across the repetitions of each auxiliary substage-game $G_{t-m}^{(t)}(s)$, for $t\geq m$ (if the associated state is visited) based on the beliefs $\beliefPi_{t-m}^{(t),-i}(s)$ and $\beliefQ_{t-m}^{(t),i}(s,\cdot)$, as described in Table \ref{algo:SFP} with a step size $\alpha_k\in[0,1]$.
They can also form beliefs about the payoff functions of these substage-games, i.e., $Q_{t-m}^{(t),i}(s,\cdot)$, and update these beliefs by taking a convex combination of the current belief and the realized one based on \eqref{eq:Qsub}, as described in Table \ref{algo:Q} with a step size $\beta_k\in[0,1]$ (which can be different from $\alpha_k$). The stochastic fictitious play for stochastic games combines the stochastic fictitious play update, described in Table \ref{algo:SFP}, and the $Q$-function update, described in Table \ref{algo:Q}, together by keeping track of $c_{\ell}^{(t)}(s)$ corresponding to the number of times the associated substage game $G_{t-m}^{(t)}(s)$ for $t\geq m$ has been played. Table \ref{algo} is a description of the dynamics from the perspective of Player $i$.

\section{Convergence Results}\label{sec:main}

In this section, I present the global convergence results for the new dynamics, described in Table \ref{algo}. To this end, we first introduce strategic equivalence of strategic-form games.

\begin{definition}\label{def:SE}
We say that two strategic-form games $G=\langle A,U \rangle$ and $\tilde{G} = \langle A,\tilde{U}\rangle$ (with the same action sets) are \textit{strategically equivalent} provided that there exists a positive constant $h^i>0$ and a function $g^i:A^{-i}\rightarrow\mathbb{R}$ such that
\be
\tilde{U}^i(a^i,a^{-i}) =  h^i \cdot U^i(a^i,a^{-i}) + g^i(a^{-i}),
\ee
for all $i\in[n]$ and $a\in A$.
\end{definition}

We focus on characterizing the convergence properties of the discrete-time dynamics based on the stochastic approximation theory through its limiting ordinary differential equation (o.d.e.) via a Lyapunov function formulation. For example, the evolution of beliefs formed about the opponent strategy in the repeated play of $G=\langle A,r\rangle$ is given by
\be\label{eq:discrete}
\beliefPi_{k+1}^{j} = \beliefPi_k^{j} + \alpha_k (a_k^j - \beliefPi_k^j),\quad\mbox{where } a_k^j \sim B^j(\beliefPi_k^{-j})
\ee
for all $j$ and the smoothed best response $B^j(\cdot)$ is as described in \eqref{eq:smoothed}. Correspondingly, if the step size satisfies the standard assumptions in the stochastic approximation methods, provided explicitly in Assumption \ref{assume:step} later, the Lipschitzness of the smoothed best response function yields that its limiting ordinary differential equation is given by
\be\label{eq:ode}
\frac{d\pi^j}{dt} = B^j(\pi^{-j}) - \pi^j,
\ee 
for all $j$ and $\pi^j:[0,\infty)\rightarrow \Delta(A^j)$ is a function of time \citep{ref:Hofbauer02}. The reader can refer to \citep{ref:Borkar08} for further information on stochastic approximation theory.

Note that the smoothed best response is a continuous function of the payoff functions though this dependence is implicit in its definition \eqref{eq:smoothed}. Therefore, the discrete-time update \eqref{eq:discrete} would have the same limiting o.d.e. \eqref{eq:ode} if players played a sequence of strategic-form games $G_k=\langle A,r_k \rangle$ such that $r_k(a)\rightarrow r(a)$ as $k\rightarrow\infty$ almost surely for all action profile $a$, i.e., $G_k\rightarrow G$. Particularly, we can view the difference between smoothed best responses in $G_k$ and $G$ as an error and this error term is asymptotically negligible since they are continuous functions of the payoffs and $G_k\rightarrow G$ as $k\rightarrow \infty$.

Stochastic fictitious play is known to converge equilibrium in two-player zero-sum games and $n$-player identical-interest games \citep{ref:Hofbauer05}. Particularly, given a strategic-form game $G=\langle A,U\rangle$, we say that $G$ is
\begin{itemize}
\item a zero-sum game if $\sum_{i\in[n]} U^i(a) = 0$ for all $a$,
\item an identical-interest game if $U^i(a)=U^{j}(a)$ for all $i,j$, and $a$.
\end{itemize}
Sufficiently smooth perturbation (i.e., sufficiently large $\tau>0$) plays an important role in characterizing the convergence properties of the discrete-time dynamics via a Lyapunov function formulation on its limiting ordinary differential equation, e.g., to ensure that the game with perturbed payoffs has isolated equilibria in identical-interest games (e.g., see \citep[Theorem 6]{ref:Hofbauer02}).

Consider a directed graph $\mathcal{G}$ whose nodes correspond to states and a directed edge exists in-between nodes associated with states $s$ and $s'$ if and only if $p(s'|s,a)>0$ for some action profile $a$. We make the following assumption about the transition probabilities across states similar to the irreducibility of Markov chains.

\begin{assumption}\label{assume:MC}
The graph $\mathcal{G}$ is connected, i.e., there exists a path in-between any two nodes.
\end{assumption}

The assumption that every state gets visited infinitely often is used, e.g., in \citep{ref:Sayin20,ref:Leslie20}, which is also implied by Assumption \ref{assume:MC} due to the smoothed best response of players ensuring that each action profile gets played with some positive probability. Furthermore, we can check Assumption \ref{assume:MC} independent of the play based on the state transition probabilities only. 

We also make the following assumption about the step sizes $\alpha_k$ and $\beta_k$ used in the updates, described in Tables \ref{algo:SFP} and \ref{algo:Q}.  

\begin{assumption}\label{assume:step}
The step sizes $\{\alpha_k\in[0,1]\}_{k\geq 0}$ and $\{\beta_k\in[0,1]\}_{k\geq 0}$ satisfy the standard conditions: $\alpha_k\rightarrow 0$ and $\beta_k\rightarrow 0$ as $k\rightarrow 0$, $\sum_{k=0}^{\infty}\alpha_k = \infty$ and $\sum_{k=0}^{\infty} \beta_k = \infty$, and $\sum_{k=0}^{\infty}\alpha_k^2 < \infty$.
\end{assumption}

We do not assume that $\sum_{k=0}^{\infty}\beta_k^2 < \infty$ since the update rule, described in Table \ref{algo:Q}, does not include a stochastic approximation error in the model-based case where state transition probabilities are known by the players. A generalization to the model-free case is discussed later in Section \ref{sec:extension}.

The following theorem provides a global convergence guarantee for the dynamics presented.

\begin{theorem}\label{thm:main}
Consider an $n$-player stochastic game, characterized by $\langle S,A,r,p,\gamma\rangle$, with turn-based controllers. Suppose that stage-payoffs induces a zero-sum or identical-interest game. Furthermore, let each player follow the learning dynamics described in Table \ref{algo}. 
Then, under Assumptions \ref{assume:MC} and \ref{assume:step}, for any $\ell=0,1,\ldots$ and $m=0,\ldots,\ell$, we have
\be\nn
(\beliefPi_{t-m}^{(t),i})_{i\in [n]} \rightarrow (\pi_m^i)_{i\in [n]},\quad\mbox{and}\quad
(\beliefQ_{t-m}^{(t),i})_{i\in [n]} \rightarrow (Q_m^i)_{i\in [n]},
\ee
as $t\rightarrow \infty$ almost surely, where $\{\pi_k^1,\ldots,\pi_k^n\}_{k=0}^{\ell}$ and $\{Q_k^1,\ldots,Q_k^n\}_{k=0}^{\ell}$ are, respectively, a Nash distribution and the corresponding $Q$-functions of a finite-horizon version of the underlying stochastic game with $\ell+1$ stages. Particularly, for each $k=0,\ldots,\ell$, the strategy profile $\pi_k=(\pi_k^i(s))_{i\in [n],s\in S}$ satisfies the condition
\be\nn
\pi^i_k(s) = \argmax_{\mu^i\in \Delta(A^i)}\left\{\Expected_{(a^i,a^{-i})\sim (\mu^i, \pi^{-i}_k(s))}[Q^i_k(s,a)] + \tau \eta^i(\mu^i)\right\}
\ee
for all $i,s$, and the associated $Q$-functions $\{Q_k\}_{k=0}^{\ell}$ are as descibed in \eqref{eq:Qsub}.\footnote{The expectation in \eqref{eq:Qsub} is now taken with respect to the randomness induced from the state transitions and the strategy profile $\{\pi_k\}_{k=0}^{\ell}$.}
\end{theorem}

\begin{proof}
The proof follows from backward induction based on the observation that the last substage-games at each epoch are stationary, i.e., have the same payoff functions across epochs, and therefore, their learning dynamics are uncoupled from the dynamics at other substage games. 
To this end, we can show that the dynamics for the last substage-game converges to a Nash distribution as it is either a zero-sum or identical-interest game. Then we can show that auxiliary substage games that are one substage before the last one converge to a game that is strategically equivalent to either a zero-sum or identical-interest game by Definitions \ref{def:turnbased} and \ref{def:SE}. Therefore, the dynamics there would also converge to a Nash distribution. We can conclude the convergence of the dynamics at any substage game by induction.  

The payoff functions of the last substage-games are stationary and equal to the stage-payoffs. Correspondingly, they are either a zero-sum or an identical-interest game based on the statement of Theorem \ref{thm:main}. Therefore, stochastic fictitious play dynamics can converge to a Nash distribution since Assumption \ref{assume:MC} yields that each auxiliary substage-game gets played infinitely often. 

Particularly, the nodes of the graph $\mathcal{G}$ (constructed based on the state transition probabilities) can form disjoint subsets $N_1,\ldots,N_d$ for some $d>1$ such that there are edges only from $N_k$ to $N_{k+1\pmod d}$ leading to a periodicity in the chain of states. However, even in such a case, each auxiliary substage game $\{G_{t-m}^{(t)}(s)\}_{t\geq m}$ specific to $m=0,1,\ldots$ and $s\in S$ would still be played infinitely often since the number of substages in epoch $t$ is $t$, i.e., it increases one by one in time. 

Since the last substage-games $G_{t}^{(t)}(s) = \langle A, r(s)\rangle$ is either zero-sum or identical-interest game and get played infinitely often, we have
\be\label{eq:pilast}
(\beliefPi_t^{(t),i}(s))_{i\in[n]}\rightarrow (\pi^i_{\ell}(s))_{i\in[n]},
\ee
as $t\rightarrow \infty$ almost surely, where $(\pi_{\ell}^i(s))_{i\in[n]}$ is a Nash distribution of the last substage-game $\langle A, r(s)\rangle$ in a finite horizon version of the underlying stochastic game with $\ell+1$ stages and $\ell=0,1,\ldots$ is arbitrary. This yields that
\be\label{eq:vlast}
\beliefV_{t}^{(t),i}(s) \rightarrow \Expected_{a\sim \pi_{\ell}(s)}[Q_{\ell}^i(s,a)],
\ee
as $t\rightarrow\infty$ almost surely for each $s\in S$, since we have $Q_{\ell}^i(s,a)=r^i(s,a)$ for all $(s,a)$. 

Suppose that for each $s\in S$ and some $m=1,\ldots,\ell$, we have
\be\label{eq:vm}
\beliefV_{t-m+1}^{(t),i}(s) \rightarrow \Expected_{a\sim \pi_{\ell-m+1}(s)}[Q_{\ell-m+1}^i(s,a)]
\ee
as $t\rightarrow\infty$ almost surely. Then, the payoff functions of the substage-games that are $m$ substages before the last ones, i.e., $G_{t-m}^{(t)}(s)=\langle A,\beliefQ_{t-m}^{(t)}(s)\rangle$, evolve according to
\begin{align}
\beliefQ_{t-m+1}^{(t+1)}&(s,a) = \left(1-\beta_{c_{t-m}^{(t)}(s)}\right) \beliefQ_{t-m}^{(t)}(s,a) 
+\beta_{c_{t-m}^{(t)}(s)}\left(r^i(s,a) + \gamma^i \sum_{s'} p(s'|s,a) \beliefV_{t-m+1}^{(t),i}(s')\right)\label{eq:Qv}
\end{align}
when $s$ gets visited at substage $t-m$ of epoch $t$. Assumption \ref{assume:MC} ensures that the update \eqref{eq:Qv} takes place infinitely often. Assumption \ref{assume:step} ensures that the step size $\beta_k$ decays sufficiently slowly such that $\beliefQ_{t-m}^{(t),i}(s,a)$ can converge to the limit of $r^i(s,a) + \gamma^i \sum_{s'} p(s'|s,a) \beliefV_{t-m+1}^{(t),i}(s')$, whose existence follows from \eqref{eq:vm}. We denote its limit by
\begin{align}
Q_{\ell-m}^i(s,a) &:= r^i(s,a)
+ \gamma^i \sum_{s'} p(s'|s,a) \Expected_{a'\sim \pi_{\ell-m+1}(s')}[Q_{\ell-m+1}^i(s',a')].\nn
\end{align}
Then, we have
\be\label{eq:Qm}
\beliefQ_{t-m}^{(t),i}(s,a) \rightarrow Q_{\ell-m}^i(s,a),
\ee
as $t\rightarrow\infty$ almost surely for each $s\in S$.
Furthermore, \eqref{eq:Qm} implies that $G_{t-m}^{(t)}(s)\rightarrow \langle A, Q_{\ell-m}^i(s,\cdot)\rangle$ as $t\rightarrow\infty$ almost surely.

Since $p(s'|s,a) = p(s'|s,a^{i_s})$ where Player $i_s$ is the only player controlling the state transitions at state $s$, as described in Definition \ref{def:turnbased}, we can write $Q_{\ell-m}^i(s,a)$ as
\begin{align}
Q_{\ell-m}^i(s,a)=r^i(s,a^{i_s},a^{-i_s}) + g^i_{\ell-m}(s,a^{i_s}),\label{eq:hg}
\end{align}
for some function $g^i_{\ell-m}(\cdot)$ that depends only on the action of Player $i_s$. Hence, the game $\langle A,Q_{\ell-m}^i(s,\cdot)\rangle$ is strategically equivalent to $\langle A, \tilde{Q}_{\ell-m}^i(s,\cdot)\rangle$ whose payoffs are defined based on the structure of the game induced by the stage payoffs $\{r^i(s,\cdot)\}_{i\in [n]}$.
\begin{itemize}
\item If $\langle A, r^i(s,\cdot)\rangle$ is a two-player zero-sum game, i.e., $r^1(s,a)+r^2(s,a) = 0$ for all $a\in A$, then we define
\begin{subequations}
\begin{align}
&\tilde{Q}_{\ell-m}^{i_s}(s,a) := r^{i_s}(s,a) + g^{i_s}_{\ell-m}(s,a^{i_s})\\
&\tilde{Q}_{\ell-m}^{-i_s}(s,a) := r^{-i_s}(s,a) - g^{i_s}_{\ell-m}(s,a^{i_s})
\end{align}
\end{subequations}
such that $\tilde{Q}_{\ell-m}^{i_s}(s,a) + \tilde{Q}_{\ell-m}^{-i_s}(s,a)=0$ for all $a\in A$.
\item If $\langle A, r^i(s,\cdot)\rangle$ is a $n$-player identical-interest game, i.e., $r^i(s,a)=r^j(s,a)$ for all $i,j\in [n]$, then we define
\begin{subequations}
\begin{align}
&\tilde{Q}_{\ell-m}^{i_s}(s,a) := r^{i_s}(s,a) + g^{i_s}_{\ell-m}(s,a^{i_s})\\
&\tilde{Q}_{\ell-m}^{j}(s,a) := r^{j}(s,a) + g^{i_s}_{\ell-m}(s,a^{i_s})\quad\forall j\neq i_s
\end{align}
\end{subequations}
such that $\tilde{Q}_{\ell-m}^{j}(s,a) = \tilde{Q}_{\ell-m}^{i_s}(s,a)$ for all $j\neq i_s$ and $a\in A$.
\end{itemize}
Since $\langle A,Q_{\ell-m}^i(s,\cdot)\rangle$ is strategically equivalent to either a two-player zero-sum or a $n$-player identical-interest game, the stochastic fictitious play across $\{G_{t-m}^{(t)}(s)\}_{t>1}$ converges to a Nash distribution of $\langle A,Q_{\ell-m}^i(s,\cdot)\rangle$ and we obtain
\be
(\beliefPi_{t-m}^{(t),i}(s))_{i\in[n]}\rightarrow (\pi^i_{\ell-m}(s))_{i\in[n]},
\ee
as $t\rightarrow \infty$ almost surely, where $(\pi_{\ell-m}^i(s))_{i\in[n]}$ is a Nash distribution of the substage-game that is $m$ substages before the last one in the finite horizon version of the underlying stochastic game. This yields that
\be
\beliefV_{t-m}^{(t),i}(s) \rightarrow \Expected_{a\sim \pi_{\ell-m}(s)}[Q_{\ell-m}^i(s,a)],
\ee
as $t\rightarrow\infty$ almost surely for each $s\in S$. Then, the proof follows from the induction.
\end{proof}

\section{Discussion}\label{sec:extension}

In this section, we discuss several research directions to extend the framework to more general cases such as: dynamics with finite memory, model-free dynamics, payoff-based dynamics, and asynchronous dynamics.

\textbf{Dynamics with Finite Memory:} In the dynamics presented, we consider that players have perfect recall. A generalization to the cases with bounded memory can be possible if players impose a hard constraint on the number of substages in an epoch. However, this can only guarantee convergence to a neighborhood of Nash distribution while the approximation error would depend on the size of their memory. On the other hand, a new variant of the dynamics with receding horizon is an interesting future research direction.

\textbf{Model-free Dynamics:} In the dynamics, we also consider that players know their own stage-payoff functions and state transition probabilities. A generalization to model-free cases where players do not know stage-payoffs and state transition probabilities is possible through a $Q$-learning update rule (introduced by \cite{ref:Watkins92}) in place of the update rule described in Table \ref{algo:Q} similar to the framework in \citep{ref:Sayin20}. This generalization brings in a stochastic approximation error in the update of $Q$-function estimates and we can impose additionally the condition $\sum_{k=0}^{\infty}\beta_k^2 < \infty$ to guarantee its almost sure convergence.

\textbf{Payoff-based Dynamics:} The dynamics necessitate access to opponent actions. Its generalization to payoff-based dynamics can be possible if players follow payoff-based dynamics for repeated play of the same game, such as individual-$Q$ learning introduced by \cite{ref:Leslie05} and the actor-critic version of the generalized weakened fictitious play introduced in \cite[Section 5]{ref:Leslie06}, in place of stochastic fictitious play.

\textbf{Asynchronous Dynamics:} A key property of the dynamics studied is to divide the horizon into epochs with finitely many stages. However, there is an implicit assumption that players are coordinated to play a finite horizon version of the underlying stochastic games within these epochs synchronously. Another interesting future research direction is to relax this coordination and to address asynchronous dynamics letting different players choose different epoch lengths.  

\section{Conclusion}\label{sec:conclusion}

In this paper, I presented a new variant of stochastic fictitious play combined with $Q$-function version of value iteration for learning in general-sum stochastic games. The dynamics presented has global convergence guarantees in stochastic games in which stage-payoffs induce a zero-sum or identical-interest game and there is only one (possibly state-dependent) player controlling the state transitions. This includes stochastic games in which stage-payoffs can have completely different structures at different states. For example, they can be identical in some and be completely opposite in others. The dynamics presented can also have convergence guarantees for player-dependent discount factors, which can lead to a non-zero-sum stochastic game even when stage payoffs sum to zero. 

One of the key properties of the dynamics is to divide the infinite horizon into epochs with finitely many stages such that players play a finite horizon version of the underlying stochastic game repeatedly. This provides players opportunities to revise and improve their beliefs formed about the play. There are several interesting future research directions discussed in Section \ref{sec:extension}.

\end{spacing}

\begin{spacing}{1}
\bibliographystyle{plainnat}
\bibliography{mybib}
\end{spacing}

\end{document}